\documentclass[final,1p,times]{elsarticle}

\usepackage{amsmath,amsthm,amssymb}

\def\R{\mathbb R}
\def\N{\mathbb N}
\def\C{\mathbb C}

\renewcommand{\d}{{\mathrm d}}
\renewcommand{\i}{{\mathrm i}}
\def\e{\mathrm e}
\def\U{\mathrm{U}}
\def\rank{\mathrm{rank}}
\def\diag{\mathrm{diag}}
\def\S{\mathcal S}
\def\Tr{\mathrm{Tr}}
\def\MR{\mathcal{M}^{\R}}
\def\MC{\mathcal M}
\def\Ker{\mathrm{Ker}}

\newtheorem{lemma}{Lemma}[section]
\newtheorem{proposition}[lemma]{Proposition}
\newtheorem{theorem}[lemma]{Theorem}
\newtheorem{corollary}[lemma]{Corollary}
\newtheorem{observation}[lemma]{Observation}
\theoremstyle{definition}
\newtheorem{definition}[lemma]{Definition}
\newtheorem{example}[lemma]{Example}
\theoremstyle{remark}
\newtheorem{remark}[lemma]{Remark}
\newtheorem{notation}[lemma]{Notation}

\makeatletter
\def\ps@pprintTitle{%
     \let\@oddhead\@empty
     \let\@evenhead\@empty
     \def\@oddfoot{\footnotesize\itshape Linear Algebra and its Applications 469 (2015) 569--593\hfill March 15, 2015}
     \let\@evenfoot\@oddfoot}
\makeatother

\begin{document}

\begin{frontmatter}

\title{Hermitian unitary matrices with modular permutation symmetry}
%
%\author[label1]{Ond\v{r}ej Turek} \ead{ondrej.turek@kochi-tech.ac.jp}
%\author[label1]{Taksu Cheon} \ead{taksu.cheon@kochi-tech.ac.jp}
\author{Ond\v{r}ej Turek} \ead{ondrej.turek@kochi-tech.ac.jp}
\author{Taksu Cheon} \ead{taksu.cheon@kochi-tech.ac.jp}

%
%\address[label1]{Laboratory of Physics, Kochi University of Technology
%Tosa Yamada, Kochi 782-8502, Japan}
\address{Laboratory of Physics, Kochi University of Technology,
Tosa Yamada, Kochi 782-8502, Japan}

\date{\today}

\begin{abstract}
We study Hermitian unitary matrices $\mathcal{S}\in\mathbb{C}^{n,n}$ with the following property: There exist $r\geq0$ and $t>0$ such that the entries of $\mathcal{S}$ satisfy $|\mathcal{S}_{jj}|=r$ and $|\mathcal{S}_{jk}|=t$ for all $j,k=1,\ldots,n$, $j\neq k$. We derive necessary conditions on the ratio $d:=r/t$ and show that these conditions are very restrictive except for the case when $n$ is even and the sum of the diagonal elements of $\S$ is zero. Examples of families of matrices $\S$ are constructed for $d$ belonging to certain intervals.
The case of real matrices $\mathcal{S}$ is examined in more detail. It is demonstrated that a real $\mathcal{S}$ can exist only for $d=\frac{n}{2}-1$, or for $n$ even and $\frac{n}{2}+d\equiv1\pmod 2$. We provide a detailed description of the structure of real $\mathcal{S}$ with $d\geq\frac{n}{4}-\frac{3}{2}$, and derive a sufficient and necessary condition of their existence in terms of the existence of certain symmetric $(v,k,\lambda)$-designs. We prove that there exist no real $\mathcal{S}$ with $d\in\left(\frac{n}{6}-1,\frac{n}{4}-\frac{3}{2}\right)$.
A parametrization of Hermitian unitary matrices is also proposed, and its generalization to general unitary matrices is given.
At the end of the paper, the role of the studied matrices in quantum mechanics on graphs is briefly explained.
\end{abstract}

%% PACS 2008
%% 03.65.Vf Phases: geometric; dynamic or topological
%% 03.65.Ca Formalism
%% 42.50.Dv Quantum state engineering and measurements (see also
%%          03.65.Ud Entanglement and quantum nonlocality, e.g.,
%%          EPR paradox, Bells inequalities, GHZ states, etc.)
%% cf. 03.67.Ac Quantum algorithms, protocols, and simulations

\end{frontmatter}

%%%
%%% Start of the main text
%%%
\
\section{Introduction}

Unitary matrices with various special properties emerge in a wide scale of applications in physics and in the engineering, and at the same time they constantly attract the attention of pure mathematicians. One of the most fascinating and longest-standing problems in mathematics is the Hadamard conjecture: \emph{If $n$ is a multiple of $4$, then there exists an $n\times n$ matrix $H$ with entries from $\{-1,1\}$ such that $HH^T=nI$.} Although the conjecture is believed to be true, no proof has yet been found. The matrix $H$ with these properties is called Hadamard matrix of order $n$, and is just a multiple of an orthogonal matrix having all the entries of the same moduli. Hadamard matrices have numerous practical applications in coding, cryptography, signal processing, artificial neural networks and many other fields, see, e.g., the monograph~\cite{Ho}.

A similar problem is related to the existence of so-called conference matrices. A conference matrix of order $n$ is an $n\times n$ matrix $C$ with $0$ on the diagonal and $\pm1$ off the diagonal such that $CC^T=(n-1)I$. Matrices of this type are important for example in telephony~\cite{Be50} and in statistics~\cite{Ra59}, but as in the case of Hadamard matrices, there is still no definite characterization of orders $n$ for which a conference matrix exists.

Note that both Hadamard and conference matrices have these two properties:
\begin{itemize}
\item[(P1)] they are multiples of orthogonal matrices;
\item[(P2)] all their off-diagonal entries are of the same moduli, and also all their diagonal entries are of the same moduli.
\end{itemize}
These properties can serve as an inspiration to generalize Hadamard and conference matrices to the whole set of matrices satisfying (P1) and (P2). A subclass fulfilling a certain additional condition, namely the class of matrices with constant diagonal, has been studied in~\cite{SL82,La83}.

Both Hadamard and conference matrices are by definition real, but they can be naturally generalized to complex ones by allowing their entries to take any values from the unit circle instead of
$\{1,-1\}$. Complex Hadamard and conference matrices and their properties are nowadays widely studied as well, see, e.g.,~\cite{Ho,Di09}. This fact may serve as another inspiration for generalizations: Examine all unitary matrices satisfying (P2).

The subject to be discussed in this paper is close to the aforementioned generalization. We will study complex unitary matrices satisfying (P2) that are also Hermitian. 
Our aim is to examine their existence and their properties, and perhaps to motivate a more extensive study of them,
as they play an important role in the quantum mechanics on graphs (we will devote Section~\ref{QuantumGraphs} at the end of the paper to a more detailed explanation). Since the real matrices of this type are for many reasons interesting, we will focus on the real case in a separate section.
Another purpose of the paper is to propose a parametrization of unitary matrices, with a particular accent put on their Hermitian subset.

\section{Preliminaries}

\begin{definition}\label{MPS}
\textit{(i)}\quad A square matrix $M\in\C^{n,n}$ is \emph{permutation-symmetric} if there are $a,b\in\C$ such that the entries of $M$ satisfy
$$
M_{jj}=a \quad \text{and} \quad M_{jk}=b \quad \text{for all $j,k=1,\ldots,n$, $j\neq k$}\,.
$$
\textit{(ii)}\quad A square matrix $M\in\C^{n,n}$ is \emph{modularly permutation-symmetric} if there are $a,b\geq0$ such that the entries of $M$ satisfy
$$
|M_{jj}|=a \quad \text{and} \quad |M_{jk}|=b \quad \text{for all $j,k=1,\ldots,n$, $j\neq k$}\,.
$$
\end{definition}

``Modularly permutation-symmetric'' will be hereinafter abbreviated as MPS. If $M$ is a permutation-symmetric matrix (or an MPS matrix) and $P$ is a permutation matrix of the same size, then $PMP^{-1}$ is a permutation-symmetric matrix (or an MPS matrix, respectively) as well.

In this paper we are particularly interested in \emph{unitary} and at the same time \emph{Hermitian} modularly permutation-symmetric matrices; we will denote them by the symbol $\S$.
As diagonal Hermitian unitary MPS matrices are trivially of the form $\S=\diag(\pm1,\pm1,\ldots,\pm1)$, from now on we will focus on the case when the modulus of the off-diagonal entries is nonzero.
For the sake of brevity, let us denote the set of all Hermitian unitary MPS matrices with the ratio $d:=\frac{|\text{diagonal entry}|}{|\text{off-diagonal entry}|}$ by the symbol $\MC_n(d)$, i.e.,
$$
\MC_n(d)=\left\{\S\in\U(n)\ \left|\ \text{$\S$ is MPS}\;\text{ and }\;\frac{|\S_{jj}|}{|\S_{jk}|}=d\;\text{ and }\;\S=\S^*\right.\right\}\,,
$$
in other words, elements of $\MC_n(d)$ are Hermitian unitary matrices $n\times n$ of the type
$$
\S=\frac{1}{\sqrt{d^2+n-1}}\left(\begin{array}{ccccc}
\pm d & \e^{\i\alpha_{12}} & \e^{\i\alpha_{13}} & \cdots & \e^{\i\alpha_{1n}} \\
\e^{-\i\alpha_{12}} & \pm d & \e^{\i\alpha_{23}} & \cdots & \e^{\i\alpha_{2n}} \\
\e^{-\i\alpha_{13}} & \e^{-\i\alpha_{23}} & \pm d & \cdots & \e^{\i\alpha_{3n}} \\
\vdots & \vdots & & \ddots & \vdots \\
\e^{-\i\alpha_{1n}} & \e^{-\i\alpha_{2n}} & \e^{-\i\alpha_{3n}} & \cdots & \pm d
\end{array}\right)\,.
$$

\begin{remark}
For $d=0$ and $d=1$, $\MC_n(d)$ represents the set of $n\times n$ Hermitian conference matrices and Hermitian Hadamard matrices, respectively:
\begin{itemize}
\item $\S\in\MC_n(0)$ if and only if $C:=\sqrt{n-1}\cdot\S$ is a (complex) Hermitian conference matrix;
\item $\S\in\MC_n(1)$ if and only if $H:=\sqrt{n}\cdot\S$ is a (complex) Hermitian Hadamard matrix.
\end{itemize}
\end{remark}

Within each set $\MC_n(d)$ we introduce an equivalence:
\begin{definition}\label{equiv}
We say that matrices $\S_1,\S_2\in\MC_n(d)$ are \emph{equivalent}, written as $\S_1\sim\S_2$, if one can be obtained from the other by performing a finite sequence of the following operations:
\begin{itemize}
\item for certain $j,k$, transpose the $j$-th and the $k$-th row, and at the same time transpose the $j$-th and the $k$-th column;
\item for certain $j$ and $\phi\in\R$, multiply the $j$-th row by $\e^{\i\phi}$, and at the same time multiply the $j$-th column by $\e^{-\i\phi}$;
\item multiply the whole matrix by $-1$.
\end{itemize}
\end{definition}
In other words, $\S_1\sim\S_2$ if and only if there exist a permutation matrix $P$ and a diagonal unitary matrix 
$D=\diag(\e^{\i\phi_1},\e^{\i\phi_2},\ldots,\e^{\i\phi_n})$ such that
$$
\S_1=DP\S_2P^{-1}D^{-1} \qquad\text{or}\qquad \S_1=-DP\S_2P^{-1}D^{-1}\,.
$$

\begin{remark}
In the literature on Hadamard matrices, a weaker equivalence is mostly used, namely that the operations can be performed independently on the rows and on the columns. Within the set $\MC_n(d)$, however, we require the equivalence as it is defined above, mainly because it ensures the property 
$(\S_1\in\MC_n(d)\,\text{ and }\,\S_1\sim\S_2)\ \Rightarrow\ \S_2\in\MC_n(d)$.
\end{remark}

\begin{notation}\label{IJ}
Everywhere in the paper, the symbols $I_k$ and $J_k$ denote the identity matrix of order $k$ and the matrix $k\times k$ all of whose entries are $1$, respectively.
\end{notation}

Finally, let us give the definition of the symmetric $(v,k,\lambda)$-design which will be useful for contructions of matrices $\S\in\MC_n(d)$ in Section~\ref{Construction} and at the end of Section~\ref{Real}.
\begin{definition}\label{SymDesign}
Let $v>k>\lambda\geq1$ be integers. A \emph{symmetric $(v,k,\lambda)$-design} is a pair $\mathcal{D}=(\mathcal{P},\mathcal{B})$, where $\mathcal{P}=\{p_1,\ldots,p_v\}$ is a set of $v$ points and $\mathcal{B}=\{B_1,\ldots,B_v\}$ is a set of $v$ subsets of $\mathcal{P}$ (blocks) each containing $k$ points, such that each pair of distinct points is contained in exactly $\lambda$ blocks.

An \emph{incidence matrix} $A=(A_{ij})$ of $\mathcal{D}$ is a $v\times v$ matrix with entries from $\{0,1\}$, where $A_{ij}=1$ if and only if $p_j\in B_i$.
\end{definition}
An $A\in\{0,1\}^{v,v}$ is an incidence matrix of a symmetric $(v,k,\lambda)$-design if and only if
\begin{equation}\label{IncidMat}
AA^T=(k-\lambda)I_v+\lambda J_v \quad\text{and}\quad AJ_v=k J_v\,,
\end{equation}
cf.~\cite{Wa}, Thm. 2.8, or~\cite{Ho}.

%%%%%%%%%%%%%%%%%%%%%%%%%%%%%%%%%%%%%%%%%%%%%%%%%%%%%%%%%%%%%%%%%%%%%%%%%%%%%%%%%%%%%%%%%%%%%%%%%%%%%

\section{Parametrization of unitary matrices}

This section addresses the problem of parametrization of unitary matrices. The result will be useful later in this paper, but we believe that it may be generally of interest in itself. The solution we present is based on ideas from \cite{CET10} and \cite{CET10b}.
 
We begin with the case when $U\in\U(n)$ is Hermitian, and then we will generalize the parametrization to all unitary matrices. At the end of the section it will be shown that after a certain minor upgrade, the parametrization is applicable much more generally, namely to Hermitian matrices $H$ solving the equation $H^2=aI+bH$.

\begin{observation}\label{spectrum}
If a matrix $\S$ is unitary and Hermitian, then the eigenvalues of $\S$ are from the set $\{-1,1\}$.
\end{observation}
The most important result of this section follows.
\begin{theorem}\label{param S}
\begin{itemize}
\item[(i)] Let $\S$ be a Hermitian unitary matrix of order $n$. If $\S\neq\pm I_n$, then there exist an $m\in\{1,\ldots,n-1\}$, a matrix $T\in\C^{m,n-m}$ and a permutation matrix $P$ such that
\begin{equation}\label{S}
\begin{split}
\S&=
-I_n+
2P\left(\begin{array}{c} I_m \\ T^* \end{array}\right)
\left(I_m+TT^*\right)^{-1}
\left(\begin{array}{cc}
I_m & T
\end{array}\right)P^{-1} \\
%=\\
&=P\left(\begin{array}{cc}
-I_m+2\left(I_m+TT^*\right)^{-1} & 
2\left(I_m+TT^*\right)^{-1}T \\
2T^*\left(I_m+TT^*\right)^{-1} & 
-I_{n-m}+2T^*\left(I_m+TT^*\right)^{-1}T
\end{array}\right)P^{-1}\,.
\end{split}
\end{equation}
\item[(ii)] For any $m\in\{1,\ldots,n-1\}$, for any $T\in\C^{m,n-m}$ and for any permutation matrix $P$ of order $n$, the matrix $\S$ given by \eqref{S} is Hermitian unitary.
\item[(iii)] If $\S$ is given by \eqref{S}, then the columns of the matrices
$$
P\left(\begin{array}{c} I_m \\ T^* \end{array}\right) \qquad\text{and}\qquad P\left(\begin{array}{c} T \\ -I_{n-m} \end{array}\right)
$$
are eigenvectors of $\S$ corresponding to the eigenvalues $1$ and $-1$, respectively.
\end{itemize}
\end{theorem}

\begin{proof}
\textit{(i)} \quad Let $\S$ be a Hermitian unitary $n\times n$ matrix different from $\pm I_n$ and $m$ denote the multiplicity of its eigenvalue $1$. Since $\S\neq\pm I_n$, $m\neq0$ and $m\neq n$. The multiplicity of the eigenvalue $-1$ equals $n-m$, and therefore
$$
\rank(\S+I_n)=n-\dim\Ker(\S+I_n)=n-(n-m)=m\in\{1,\ldots,n-1\}\,.
$$
Hence there is an invertible $M\in\C^{m,m}$ and a permutation matrix $P$ such that
$$
\S+I_n=P\left(\begin{array}{cc}
M & MT_1 \\
T_2M & T_2MT_1
\end{array}\right)P^{-1}\,;
$$
note that $P$ can be omitted if and only if the upper left submatrix $m\times m$ of $\S+I_n$ is invertible.
As $\S=\S^*$, necessarily $M=M^*$ and $(MT_1)^*=T_2M$. Since $M$ is invertible, we have $T_2=T_1^*$. For brevity $T:=T_1$. Since $\S$ is unitary,
\begin{equation*}
\begin{split}
\S\S^*=&I_n+2P\left(\begin{array}{cc}
M & MT \\
T^*M & T^*MT
\end{array}\right)P^{-1}
+P\left(\begin{array}{cc}
M^2+MTT^*M & M^2T+MTT^*MT \\
T^*M^2+T^*MTT^*M & T^*M^2T+T^*MTT^*MT
\end{array}\right)P^{-1} \\
=&I_n\,,
\end{split}
\end{equation*}
hence we obtain $2M+M^2+MTT^*M=0$, equivalently $2M^{-1}=I_m+TT^*$. Consequently, $M=2(I_m+TT^*)^{-1}$.

\textit{(ii)} \quad Any $\S$ given by \eqref{S} obviously satisfies $\S\S^*=I_n$ and $\S=\S^*$.

\textit{(iii)} \quad If $\S$ is given by \eqref{S}, a straightforward calculation gives
$$
\S P\left(\begin{array}{c} I_m \\ T^* \end{array}\right)=P\left(\begin{array}{c} I_m \\ T^* \end{array}\right) \quad\text{and}\quad \S P\left(\begin{array}{c} T \\ -I_{n-m} \end{array}\right)=-P\left(\begin{array}{c} T \\ -I_{n-m} \end{array}\right)\,,
$$
therefore (iii) holds true.
\end{proof}

\begin{remark}\label{P}
Let $\S\neq\pm I_n$ be a Hermitian unitary matrix of order $n$, $m=\rank(\S+I_n)$, and $\S^{(1,1)}$ be the upper left $m\times m$ submatrix of $\S$. It follows from the proof of Theorem~\ref{param S} that the permutation matrix $P$ must be involved in the parametrization~\eqref{S} if and only if $\S^{(1,1)}+I_m$ is singular.
\end{remark}

\begin{remark}
Since the matrix $T$ occurring in the parametrization~\eqref{S} determines the eigenvectors of $\S$, it is related to the diagonalization of $\S$ as well. It follows from Theorem~\ref{param S}~(iii) that
$$
\S = X_m Z_m X_m^{-1}
$$
for
$$
X_m = P\left(\begin{array}{cc} I_m & T \\ T^* & -I_{n-m} \end{array}\right)\,, 
\qquad
Z_m = \left(\begin{array}{cc}I_m & 0 \\ 0 & -I_{n-m} \end{array}\right)\,.
$$
\end{remark}

The main idea of Theorem~\ref{param S} can be extended to a general unitary matrix:
\begin{theorem}\label{param U}
Let $U\in\U(n)$ such that $U\neq-I_n$. Let $n-m$ denote the multiplicity of its eigenvalue $-1$. Then
\begin{itemize}
\item[(i)] If $n-m\neq0$, then there exists a $T\in\C^{m,n-m}$, a Hermitian $S\in\C^{m,m}$ and a permutation matrix $P$ such that
\begin{equation}\label{U}
\begin{split}
U&=-I_n+
2P\left(\begin{array}{c} I_m \\ T^* \end{array}\right)
\left(I_m+TT^*+\i S\right)^{-1}
\left(\begin{array}{cc}
I_m & T
\end{array}\right)P^{-1} \\
&=-I_n+
2P\left(\begin{array}{cc}
\left(I_m+TT^*+\i S\right)^{-1} & 
\left(I_m+TT^*+\i S\right)^{-1}T \\
T^*\left(I_m+TT^*+\i S\right)^{-1} & 
T^*\left(I_m+TT^*+\i S\right)^{-1}T
\end{array}\right)P^{-1}\,,
\end{split}
\end{equation}
and conversely, any matrix given by \eqref{U} is unitary.
\item[(ii)] If $n-m=0$, there exists a Hermitian $S\in\C^{m,m}$ such that $U=-I_n+2\left(I_m+\i S\right)^{-1}$, and conversely, any $U$ given by this formula is unitary.
\end{itemize}
\end{theorem}

\begin{proof}
Let $n-m\neq0$. As in the proof of Theorem~\ref{param S}, we start from the decomposition $U+I_n=P\left(\begin{array}{c}I\\T_2\end{array}\right)M\left(\begin{array}{cc}I&T_1\end{array}\right)P^{-1}$, where $M\in\C^{m,m}$ is invertible, and then require $UU^*=I_n$. This leads to $T_2=T_1^*$ and $M=2(I+T_1T_1^*+\i S)^{-1}$ for a certain Hermitian matrix $S$.
If $n-m=0$, the matrix $U+I_n$ is invertible. Let us denote $U+I_n=:M$. Then the requirement $UU^*=I_n$ gives $M=2(I_n+\i S)^{-1}$ for a certain Hermitian $S$.
\end{proof}

\begin{remark}\label{P U}
The idea from Remark~\ref{P} applies to~\eqref{U} as well. The permutation matrix $P$ must be involved in~\eqref{U} if and only if $U^{(1,1)}+I_m$ is singular, where $U^{(1,1)}$ stands for the upper left submatrix $m\times m$ of $U$ and $m=\rank(\U+I_n)$. In case $U^{(1,1)}+I_m$ is invertible, $P$ may be omitted.
\end{remark}

\begin{remark}
The unitary group $\U(n)$ has $n^2$ real parameters. There exist several known parametrizations, i.e., ways how the parameters can be assigned to matrices $U\in\U(n)$, for example~\cite{Mu62,Di82} and many other. In accordance with P.~Di\c{t}\u{a} (cf. e.g.~\cite{Di94}), we call a parametrization \emph{natural} if the involved parameters are free, i.e., there are no supplementary restrictions upon them to enforce unitarity. Our solution~\eqref{U} falls within that class. On the other hand, \eqref{U} has a disadvantage that if the rows and columns of $U$ are not suitably ordered, then a permutation matrix must be brought in, see Remark~\ref{P U}.
\end{remark}

\subsection*{Hermitian solutions of quadratic matrix equations}

The reader may have observed in the proof of Theorem~\ref{param S} that the essential properties of $\S$ that allowed us to obtain the parametrization~\eqref{S} were the following two: the hermiticity of $\S$ and the fact that $\S$ has only two eigenvalues. In the light of this idea, we will generalize the parametrization~\eqref{S}, originally developped for Hermitian unitary matrices (i.e., solutions of $\S^2=I$), to Hermitian solutions of more general matrix quadratic equations
\begin{equation}\label{MatPol}
H^2=aI+bH \qquad (a,b\in\R)\,.
\end{equation}

We observe at first that the eigenvalues of any solution $H$ of~\eqref{MatPol} must satisfy $\lambda^2=a+b\lambda$, hence $\sigma(H)=\{\lambda_1,\lambda_2\}$ where $\lambda_{1,2}=\frac{1}{2}\left(b\pm\sqrt{b^2+4a}\right)$. Since $\lambda_{1,2}$ are real due to the hermiticity of $H$, one has to assume $a,b\in\R$ and $4a+b^2\geq0$. Note that the case $4a+b^2=0$ is not interesting, because it represents the situation when any Hermitian solution of~\eqref{MatPol} has the eigenvalue $b/2$ with multiplicity $n$, thus $H=\frac{b}{2}I$. For these reasons we shall assume the strict inequality $4a+b^2>0$.

Let us transform Equation~\eqref{MatPol} into its equivalent form
$$
\left(H-\frac{b}{2}I\right)^2=\left(a+\frac{b^2}{4}\right)I
$$
and define
$$
M:=\frac{2}{\sqrt{4a+b^2}}\left(H-\frac{b}{2}I\right)\,.
$$
Matrix $M$ is Hermitian (because $H$ is Hermitian) and at the same time unitary, since it satisfies $M^2=I$. Therefore we can apply Theorem~\ref{param S} and in this way obtain the sought parametrization of $H$, see Theorem~\ref{Quadratic} below. We remark that the trivial solutions of~\eqref{MatPol}, namely $H=\frac{1}{2}\left(b\pm\sqrt{4a+b^2}\right)I$, are excluded from the parametrization, just as $\S=\pm I$ have been excluded in Theorem~\ref{param S}.

\begin{theorem}\label{Quadratic}
Let $a,b\in\R$, $4a+b^2>0$.
\begin{itemize}
\item[(i)] A Hermitian $n\times n$ matrix $H$ different from $\frac{1}{2}\left(b\pm\sqrt{4a+b^2}\right)I_n$ satisfies $H^2=aI_n+bH$ if and only if
\begin{equation}\label{H}
H=
\frac{b-\sqrt{4a+b^2}}{2}I_n+
\sqrt{4a+b^2}\cdot P\left(\begin{array}{c} I_m \\ T^* \end{array}\right)
\left(I_m+TT^*\right)^{-1}
\left(\begin{array}{cc}
I_m & T
\end{array}\right)P^{-1}
\end{equation}
for an $m\in\{1,\ldots,n-1\}$, a matrix $T\in\C^{m,n-m}$ and a permutation matrix $P$.
\item[(ii)] If $H$ is given by \eqref{H}, then the columns of the matrices
$$
P\left(\begin{array}{c} I_m \\ T^* \end{array}\right) \qquad\text{and}\qquad P\left(\begin{array}{c} T \\ -I_{n-m} \end{array}\right)
$$
are eigenvectors of $H$ corresponding to the eigenvalues $\frac{b+\sqrt{4a+b^2}}{2}$ and $\frac{b-\sqrt{4a+b^2}}{2}$, respectively.
\end{itemize}
\end{theorem}

%%%%%%%%%%%%%%%%%%%%%%%%%%%%%%%%%%%%%%%%%%%%%%%%%%%%%%%%%%%%%%%%%%%%%%%%%%%%%%%%%%%

\section{Modular permutation symmetry}\label{Complex}

In the following part of the paper we will study Hermitian unitary MPS matrices. Prior to that, let us bring in a proposition characterizing the set of Hermitian unitary \emph{permutation-symmetric} matrices (cf.~\cite{ET06}):
\begin{proposition}
A unitary $n\times n$ matrix $U$ is permutation-symmetric if and only if $U=aI_n+bJ_n$ for $a,b\in\C$ satisfying $|a|=1$ and $|a+nb|=1$. Moreover, if $U$ is Hermitian, then $U=\pm(I_n-\frac{2}{n}J_n)$.
\end{proposition}
We see that only two permutation-symmetric Hermitian unitary matrices exist,
both corresponding to $d=\frac{n}{2}-1$.
However, once the permutation symmetry is weakened to the \emph{modular permutation symmetry},
there is much more freedom for $d$, as we shall see.

In this section we will examine general properties of Hermitian unitary MPS matrices, in particular necessary conditions of their existence, whereas sufficient conditions and concrete examples of such matrices will be presented in Section~\ref{Construction}.

\begin{proposition}\label{range of r}
Let $\S\in\MC_n(d)$. If $n>2$, then $d\leq\frac{n}{2}-1$.
\end{proposition}
\begin{proof}
The diagonal entries of $\S$ are $+r$ and $-r$ for $r=\frac{d}{\sqrt{d^2+n-1}}$. Since $n>2$, at least two of them are equal, we may suppose without loss of generality that $\S_{11}=\S_{22}$. Moreover, we assume $\S_{11}=+r$; alternatively we would work with the equivalent matrix $-\S$. The unitarity of $\S$ requires $[\S\S^*]_{12}=0$, where
$$
\left[\S\S^*\right]_{12}=\S_{11}\overline{\S_{21}}+\S_{12}\overline{\S_{22}}+\sum_{j=3}^{n}\S_{1j}\overline{\S_{2j}}\,.
$$
Let us denote $\S_{jk}=t\e^{\i\alpha_{jk}}$ for $t=\frac{1}{\sqrt{d^2+n-1}}$. Since $\S$ is Hermitian, it holds $\overline{\S_{21}}=\S_{12}$. Therefore, the condition $[\S\S^*]_{12}=0$ leads to
$$
2rt\e^{\i\alpha_{12}}+t^2\sum_{j=3}^{n}\e^{\i(\alpha_{1j}-\alpha_{2j})}=0\,,
$$
hence
$$
\frac{r}{t}=-\frac{\e^{-\i\alpha_{12}}}{2}\sum_{j=3}^{n}\e^{\i(\alpha_{1j}-\alpha_{2j})}\,.
$$
Consequently,
$$
d=\frac{r}{t}\leq\frac{1}{2}(n-2)=\frac{n}{2}-1\,.
$$
\end{proof}

Now we derive a relation between $d$ and the signs of the diagonal entries of $\S$.

\begin{proposition}\label{Trace}
Let $\S\in\MC_n(d)$, let $p$ denote the number of its non-negative diagonal entries, and let $m$ be the multiplicity of its eigenvalue $1$. Then
\begin{equation}\label{TraceEquality}
2m-n=(2p-n)\frac{d}{\sqrt{d^2+n-1}}\,.
\end{equation}
\end{proposition}
\begin{proof}
Since $\S\in\MC_n(d)$, its diagonal entries are $\pm\frac{d}{\sqrt{d^2+n-1}}$. According to the assumptions, $\Tr(\S)=p\frac{d}{\sqrt{d^2+n-1}}+(n-p)\left(-\frac{d}{\sqrt{d^2+n-1}}\right)=(2p-n)\frac{d}{\sqrt{d^2+n-1}}$. On the other hand, since $\S$ is unitary and at the same time Hermitian, its eigenvalues are from the set $\{1,-1\}$, see Observation~\ref{spectrum}. The multiplicity of $1$ is $m$, the multiplicity of $-1$ is $n-m$, hence $\Tr(\S)=m\cdot1+(n-m)\cdot(-1)=2m-n$. Comparing these two expressions for $\Tr(\S)$ we obtain Equation~\eqref{TraceEquality}.
\end{proof}

\begin{notation}\label{mp}
From now on to the end of the paper, the symbols $m$ and $p$ are reserved for the multiplicity of the eigenvalue $1$ and the number of non-negative diagonal elements, respectively, of matrices $\S\in\MC_n(d)$. 
\end{notation}

\begin{example}
We demonstrate the use of formula~\eqref{TraceEquality} on the extremal values of $d$, namely $d=0$ (conference matrices) and $d=\frac{n}{2}-1$.
\begin{itemize}
\item Let $d=0$. Then Equation~\eqref{TraceEquality} gives $m-2n=0$. Consequently, complex Hermitian conference matrices exist only for even $n$.
\item Let $d=\frac{n}{2}-1$. Then Equation~\eqref{TraceEquality} takes the form $2m-n=(2p-n)\left(1-\frac{2}{n}\right)$, which is equivalent to $\frac{2p}{n}=p+1-m$. Since $m$ and $p$ are integers, the quantity $\frac{2p}{n}$ must be an integer as well. Furthermore, it holds $0\leq p\leq n$ by definition of $p$. Hence we obtain three possible values of $p$: $p=0$, $p=n$, and $p=\frac{n}{2}$, the third one only for even $n$. They correspond to the following solutions of Equation~\eqref{TraceEquality}: $(m,p)=(1,0)$, $(m,p)=(n-1,n)$, and $(m,p)=(\frac{n}{2},\frac{n}{2})$. These three solutions together with the parametrization of unitary matrices \eqref{S} can be used to an easy construction of all elements of $\MC_n(\frac{n}{2}-1)$, cf. also~\cite{CT10}.
\end{itemize}
\end{example}

\begin{theorem}\label{m,p,d}
Let $\S\in\MC_n(d)$ and $m,p$ have the usual meaning (see Notation~\ref{mp}). Then either
(i) $p=m=\frac{n}{2}$
or (ii) $\bigl(p<m<\frac{n}{2}\bigr)\vee\bigl(p>m>\frac{n}{2}\bigr)$. In case (ii) it holds $d=\left|m-\frac{n}{2}\right|\sqrt{\frac{n-1}{(p-m)(p+m-n)}}$.
\end{theorem}

\begin{proof}
Any $\S\in\MC_n(d)$ satisfies Equation~\eqref{TraceEquality}, which yields the following alternative:
\begin{itemize}
\item[(i)] $2m-n=2p-n=0$, i.e., $p=m=\frac{n}{2}$.
\item[(ii)] $|2m-n|>|2p-n|>0$ and at the same time $\frac{d}{\sqrt{d^2+n-1}}=\frac{2m-n}{2p-n}$. It can be written equivalently as $d=\left|m-\frac{n}{2}\right|\sqrt{\frac{n-1}{(p-m)(p+m-n)}}$, where moreover $p$ and $m$ must satisfy $p>m>\frac{n}{2}$ or $p<m<\frac{n}{2}$.
\end{itemize}
\end{proof}

\begin{remark}
Let $\S\in\MC_n(d)$ for $d\notin\left\{\left(m-\frac{n}{2}\right)\sqrt{\frac{n-1}{(p-m)(p+m-n)}}\,,\,\frac{n}{2}<m<p\leq n\right\}$. Then, with regard to Theorem~\ref{m,p,d}, $m=p=\frac{n}{2}$, which means in particular that $n$ must be even.
\end{remark}

We finish the section with a remark on a matrices $\S$ with $m=\frac{n}{2}$.
\begin{remark}\label{normal T}
Let $\S\in\MC_n(d)$ for $m=\frac{n}{2}$.
Due to Theorem~\ref{param S}~(i), there exists a square matrix $T\in\C^{m,m}$ such that
$$
\S\sim
\left(\begin{array}{cc}
-I_m+2\left(I_m+TT^*\right)^{-1} & 
2\left(I_m+TT^*\right)^{-1}T \\
2T^*\left(I_m+TT^*\right)^{-1} & 
-I_m+2T^*\left(I_m+TT^*\right)^{-1}T
\end{array}\right)\,.
$$
Among all matrices of this type, those with normal $T$ are particularly useful, because in such a case
$-I_m+2T^*\left(I_m+TT^*\right)^{-1}T=-\left[-I_m+2\left(I_m+TT^*\right)^{-1}\right]$, and consequently
\begin{equation}\label{FG}
\S\sim\left(\begin{array}{cc}
F & G \\
G^* & -F
\end{array}\right)\,,
\end{equation}
where $F=-I_m+2\left(I_m+TT^*\right)^{-1}$ and $G=2\left(I_m+TT^*\right)^{-1}T$\,.
We will take advantage of the special form~\eqref{FG} of matrices $\S$ in the following section.
\end{remark}

%%%%%%%%%%%%%%%%%%%%%%%%%%%%%%%%%%%%%%%%%%%%%%%%%%%%%%%%%%%%%%%%%%%%%%%%%%%%%%%%%%%%

\section{Construction of Hermitian unitary MPS matrices}\label{Construction}

Let us propose several ways how matrices $\MC_n(d)$ can be constructed for certain values of $d$.

First of all, for $n=2$ there exists an $\S\in\MC_2(d)$ for any $d>0$. Moreover, $\S$ can be always chosen real: $\S=\frac{1}{\sqrt{d^2+1}}\left(\begin{array}{cc} d & 1 \\ 1 & -d \end{array}\right)$.

\medskip

\emph{From now on let $n>2$.} Now $d$ is bounded from above by $\frac{n}{2}-1$ (Prop.~\ref{range of r}). With regard to this fact, we will structure our presentation according to the value of $d$, starting from the upper bound. The proposed matrix constructions mostly satisfy $m=p=\frac{n}{2}$, and will be moreover based on the scheme~\eqref{FG} from Remark~\ref{normal T}; recall that any setting different from $m=p=\frac{n}{2}$ would lead to a significant restriction on the admissible values of $d$, see Theorem~\ref{m,p,d}.

\subsection*{Case $d=\frac{n}{2}-1$}
\emph{There exists an $\S\in\MC_n(\frac{n}{2}-1)$ for all $n>2$, and $\S$ can be chosen real.}

According to \cite{CT10}, any $\S\in\MC_n(\frac{n}{2}-1)$ is equivalent either to $I_n-\frac{2}{n}J_n$ or to the matrix obtained from~\eqref{UpperInt} below by setting $\alpha=0$.

\subsection*{Case $d\in\left[\frac{n}{2}-3,\frac{n}{2}-1\right)$}
\emph{There exists an $\S\in\MC_n(d)$ for all even $n\in\N$.} For example, set $m=\frac{n}{2}$ and
\begin{equation}\label{UpperInt}
\S=\frac{1}{\sqrt{d^2+n-1}}\left(\begin{array}{cc}
(d+1)I_m-J_m & (\e^{\i\alpha}-1)I_m+J_m \\
(\e^{-\i\alpha}-1)I_m+J_m & -(d+1)I_m+J_m
\end{array}\right)\,,
\end{equation}
where $\alpha$ is chosen so that $\cos\alpha=d+2-\frac{n}{2}$.

\subsection*{Case $d\in\left(\frac{n}{4}-\frac{3}{2},\frac{n}{2}-3\right)$}

\emph{If there exists a symmetric $(v,k,\lambda)$-design for $v=\frac{n}{2}$, then there exists an $\S\in\MC_n(d)$ for all $d\in\left[\frac{n}{2}-1-2(k-\lambda),\frac{n}{2}-1\right]$.}
The statement follows from Proposition~\ref{S incidence} below.

\begin{proposition}\label{S incidence}
Let $A$ be the incidence matrix of a symmetric $(v,k,\lambda)$-design, $\alpha\in[0,2\pi)$ and $G$ be the $v\times v$ matrix given as
$$
G_{j\ell}=\e^{\i\alpha A_{j\ell}} \qquad \text{for all } j,\ell=1,\ldots,n\,.
$$
If $d=v-1-(k-\lambda)(1-\cos\alpha)$, then
\begin{equation}\label{S G}
\S=\frac{1}{\sqrt{d^2+2v-1}}\left(\begin{array}{cc}
(d+1)I_v-J_v & G \\
G^* & -(d+1)I_v+J_v
\end{array}\right)
\end{equation}
satisfies $\S\in\MC_{2v}(d)$.
\end{proposition}
\begin{proof}
It suffices to prove that $\S\S^*=I_{2v}$. With regard to~\eqref{S G},
$$
\S\S^*=\frac{1}{d^2+2v-1}\begin{pmatrix}
(d+1)^2I_v+[v-2(d+1)]J_v+GG^* & GJ_v-J_vG \\
J_vG^*-G^*J_v & (d+1)^2I_v+[v-2(d+1)]J_v+G^*G
\end{pmatrix}\,.
$$
Therefore, we shall prove
\begin{equation}\label{GJ}
GJ_v-J_vG=0 \qquad\text{and}\qquad (d+1)^2I_v+[v-2(d+1)]J_v+GG^*=(d^2+2v-1)I_v\,.
\end{equation}
The matrix $A$, being the incidence matrix of a symmetric $(v,k,\lambda)$-design, has the following properties:
\begin{itemize}
\item[(a)] every row and every column of $A$ contains $k$ entries $+1$ and $(v-k)$ entries $0$,
\item[(b)] the multiset $\{A_{ji}-A_{\ell i}\,|\,i=1,\ldots,m\}$ ($j\neq\ell$) equals $\{\underbrace{1,\ldots,1}_{k-\lambda},\underbrace{-1,\ldots,-1}_{k-\lambda},\underbrace{0,\ldots,0}_{v-2(k-\lambda)}\}$.
\end{itemize}
From (a) it follows $[GJ_v]_{j\ell}=[J_vG]_{j\ell}=k\cdot\e^{\i\alpha}+v-k$ for all $j,\ell$, hence $GJ_v-J_vG=0$. From (b) we obtain $[GG^*]_{j\ell}=v-2(k-\lambda)(1-\cos\alpha)$ for all $j\neq\ell$. Since $[GG^*]_{jj}=v$, we have $GG^*=vI_v+\left[v-2(k-\lambda)(1-\cos\alpha)\right](J_v-I_v)$. These facts together with the assumption $d=v-1-(k-\lambda)(1-\cos\alpha)$ prove equations~\eqref{GJ}, hence $\S\S^*=I_{2v}$.
\end{proof}

If $\alpha$ runs over $[0,2\pi)$, the quantity $d=v-1-(k-\lambda)(1-\cos\alpha)$ attains all values in the interval $\left[v-1-2(k-\lambda),v-1\right]$. Consequently, if a symmetric $(v,k,\lambda)$-design for $v=\frac{n}{2}$ is known, then 
Proposition~\ref{S incidence} allows to construct matrices $\S\in\MC_n(d)$ for all $d\in\left[\frac{n}{2}-1-2(k-\lambda),\frac{n}{2}-1\right]$.

Note, however, that for any symmetric $(v,k,\lambda)$-design, the value $k-\lambda$ is bounded from above by $\frac{v+1}{4}$ (see, e.g.,~\cite{An97}, Thm. 3.1.2). Therefore, the construction~\eqref{S G} works only for $d\geq\frac{n}{4}-\frac{3}{2}$. The value $k-\lambda$ attains the maximum $\frac{v+1}{4}$ for a symmetric $(v,\frac{v}{2}-\frac{1}{2},\frac{v}{4}-\frac{3}{4})$-design.
Such combinatorial design is called \emph{Hadamard design of order $\frac{v+1}{4}$}, and its existence is equivalent to the existence of an Hadamard matrix of order $v+1$ (\cite{An97}, Thm. 3.2.4). Hence we get the following corollary of Proposition~\ref{S incidence}.

\begin{corollary}
If there exists an Hadamard matrix of order $\frac{n}{2}+1$, then there exists an $\S\in\MC_n(d)$ for all $d\in\left[\frac{n}{4}-\frac{3}{2},\frac{n}{2}-1\right]$.
\end{corollary}

The required incidence matrix of a Hadamard design, used for constructing $\S$ (cf. Proposition~\ref{S incidence}), can be obtained by the formula $A=\frac{1}{2}\left(K_H+J_\frac{n}{2}\right)$, where $K_H$ is the \emph{core} of an Hadamard matrix of order $\frac{n}{2}+1$; see the definition below.

\begin{definition}
\textit{(i)}\quad Let $H$ be an Hadamard matrix of order $N$ having the form
\begin{equation}\label{standardH}
H=\left(\begin{array}{c|ccc}
1 & 1 & \cdots & 1 \\
\hline
1 & \\
\vdots & & K_H \\
1 &
\end{array}\right)\,.
\end{equation}
The $(N-1)\times(N-1)$ matrix $K_H$ is called the core of the Hadamard matrix $H$.

\textit{(ii)}\quad Let $C$ be a conference matrix of order $N$ having the form
\begin{equation}\label{standardC}
C=\left(\begin{array}{c|ccc}
0 & 1 & \cdots & 1 \\
\hline
1 & \\
\vdots & & K_C \\
1 & \\
\end{array}\right)\,.
\qquad
\end{equation}
The $(N-1)\times(N-1)$ matrix $K_C$ is called the core of the conference matrix $C$.
\end{definition}
It is easy to see that if an Hadamard matrix of order $N$ exists, then an Hadamard matrix of the form~\eqref{standardH} exists. Similarly, if there is a conference matrix of order $N$, then there is a conference matrix having the form~\eqref{standardC}.

\subsection*{Case $d=\frac{n}{4}-\frac{3}{2}$}
\emph{There exists an $\S\in\MC_n(\frac{n}{4}-\frac{3}{2})$ for any even $n$.}

An $\S\in\MC_n(\frac{n}{4}-\frac{3}{2})$ can be constructed as follows. Let $K_H$ be a core of a complex Hadamard matrix of order $\frac{n}{2}+1$, i.e., $K_H$ is of the size $\frac{n}{2}\times\frac{n}{2}$. Note that there exists a complex Hadamard matrix of any order $N$, e.g., the one given by $H_{jk}=\e^{2\pi\i(j-1)(k-1)/N}$. Then
\begin{equation}\label{S from core}
\S=\frac{1}{\sqrt{d^2+n-1}}\left(\begin{array}{cc}
(\frac{n}{4}-\frac{1}{2})I_m-J_m & K_H \\
K_H^* & -(\frac{n}{4}-\frac{1}{2})I_m+J_m
\end{array}\right)
\end{equation}
with $m=\frac{n}{2}$ satisfies $\S\in\MC_n(\frac{n}{4}-\frac{3}{2})$.

\subsection*{Case $d\in\left[\frac{n}{4}-\frac{3}{2}-\frac{1}{n-2},\frac{n}{4}-\frac{3}{2}\right)$}
\emph{If there exists a symmetric conference matrix of order $\frac{n}{2}+1$, then there exists an $\S\in\MC_n(d)$ for all $d\in\left[\frac{n}{4}-\frac{3}{2}-\frac{1}{n-2},\frac{n}{2}-1\right]$}, see Proposition~\ref{Conference core}.
\begin{proposition}\label{Conference core}
Let $K_C$ be a core of a symmetric conference matrix of order $\frac{n}{2}+1$ and $G_C$ be defined by $[G_C]_{jk}=\e^{\i\alpha [K_C]_{jk}}$, where $\alpha$ is chosen such that $d=\frac{n-6}{4}+\cos\alpha+\frac{n-2}{4}\cos^2\alpha$. Then
\begin{equation}\label{S G_C}
\S=\frac{1}{\sqrt{d^2+n-1}}\left(\begin{array}{cc}
(d+1)I_m-J_m & G_C \\
G_C^* & -(d+1)I_m+J_m
\end{array}\right)
\end{equation}
with $m=\frac{n}{2}$ satisfies $\S\in\MC_n(d)$ for all $d\in\left[\frac{n}{4}-\frac{3}{2}-\frac{1}{n-2},\frac{n}{2}-1\right]$. 
\end{proposition}
\begin{proof}
The proof is similar to the proof of Proposition~\ref{S incidence}. We shall show that $\S\S^*=I_n$. With regard to~\eqref{S G_C}, this condition is equivalent to
\begin{equation}\label{G_CJ}
G_CJ_m-J_mG_C=0 \qquad\text{and}\qquad (d+1)^2I_m+[m-2(d+1)]J_m+G_CG_C^*=(d^2+n-1)I_m\,.
\end{equation}
Since $K_C$ is a core of a symmetric conference matrix, the following two statements can be proven:
\begin{itemize}
\item[(a)] Every row and every column of $K_C$ contains $\frac{m-1}{2}$ entries $(+1)$, $\frac{m-1}{2}$ entries $(-1)$ and one entry $0$.
\item[(b)] The multiset $\left\{[K_H]_{ji}-[K_H]_{\ell i}\,|\,i=1,\ldots,m\right\}$ equals $\{\underbrace{2,\ldots,2}_{(m-1)/4},\underbrace{-2,\ldots,-2}_{(m-1)/4},1,-1,\underbrace{0,\ldots,0}_{(m-3)/2}\}$ for all $j\neq\ell$.
\end{itemize}
From (a) it follows $G_CJ-JG_C=0$. From (b) we obtain
$$
[G_CG_C^*]_{j\ell}=\frac{m-1}{4}(\e^{2\i\alpha}+\e^{-2\i\alpha})+\e^{\i\alpha}+\e^{-\i\alpha}+\frac{m-3}{2}=(m-1)\cos^2\alpha+2\cos\alpha-1
$$
for all $j\neq\ell$. Since $[G_CG_C^*]_{jj}=m$, we have $G_CG_C^*=mI_m+\left[(m-1)\cos^2\alpha+2\cos\alpha-1\right](J_m-I_m)$. Taking the assumption $d=\frac{n-6}{4}+\cos\alpha+\frac{n-2}{4}\cos^2\alpha$ into account, we obtain the second equation~\eqref{G_CJ}. Thus, $\S\S^*=I_n$. Finally, one can easily demonstrate that if $\alpha$ runs over $[0,\pi]$, then $d$ attains all values in $\left[\frac{n}{4}-\frac{3}{2}-\frac{1}{n-2},\frac{n}{2}-1\right]$.
\end{proof}

\subsection*{Case $d\in\left[0,1\right]$}
\emph{If there is an Hermitian conference matrix $C$ of order $\frac{n}{2}$ (equivalently: $\MC_{\frac{n}{2}}(0)\neq\emptyset$), then there exists an $\S\in\MC_n(d)$ for all $d\in[0,1]$.} It is easy to check that this matrix can be constructed as
\begin{equation}\label{S C}
\S=\frac{1}{\sqrt{d^2+n-1}}\left(\begin{array}{cc}
dI_m+C & C-\e^{\i\alpha}I_m \\
C-\e^{-\i\alpha}I_m & -(dI_m+C)
\end{array}\right)\,,
\end{equation}
where $m=\frac{n}{2}$ and $\alpha$ is chosen such that $d=\cos\alpha$.

%%%%%%%%%%%%%%%%%%%%%%%%%%%%%%%%%%%%%%%%%%%%%%%%%%%%%%%%%%%%%%%%%%%%%%%%%%%%%%%%%%%%%%%%%%%%%%%%%%%%%%%%

\section{Real case (Symmetric orthogonal matrices)}\label{Real}

In this section we will focus on the matrices $\S\in\MC_n(d)$ with the additional property that all their entries are real, i.e., on symmetric orthogonal matrices of the type
$$
\frac{1}{\sqrt{d^2+n-1}}\left(\begin{array}{ccccc}
\pm d & \pm1 & \pm1 & \cdots & \pm1 \\
\pm1 & \pm d & \pm1 & \cdots & \pm1 \\
\pm1 & \pm1 & \pm d & \cdots & \pm1 \\
\vdots & \vdots & & \ddots & \vdots \\
\pm1 & \pm1 & \cdots & \pm1 & \pm d
\end{array}\right)\,.
$$
In what follows we will denote the real subset of $\MC_n(d)$ by $\MR_n(d)$, i.e.,
$$
\MR_n(d)=\left\{\S\in\U(n)\cap\R^{n,n}\ \left|\ \text{$\S$ is MPS}\;\text{ and }\;\frac{|\S_{jj}|}{|\S_{jk}|}=d\;\text{ and }\;\S=\S^T\right.\right\}\,.
$$
Elements of $\MR_n(0)$ and $\MR_n(1)$ represent (up to the factor $\frac{1}{\sqrt{d^2+n-1}}$) symmetric Hadamard and symmetric conference matrices, respectively, of order $n$. For this reason, matrices $\S\in\MR_n(d)$ with $d\in\left[0,\frac{n}{2}-1\right]$ can be regarded as a straightforward generalization of the concept of symmetric Hadamard/conference matrices. A special subset of them, namely matrices $\S\in\MR_n(d)$ with constant signs of the diagonal elements, have been studied in~\cite{SL82}. In this section we are interested in the case with general, mixed diagonal signs.

Let us begin with examination of matrices of small orders. If $n$ is small, it is an easy excercise to find admissible values of $d$ using the orthogonality of the matrix rows. The results are summarized in the following Observation.
\begin{observation}\label{small n}
The following hold:
\begin{itemize}
\item $\MR_2(d)$ is non-empty for all $d\in[0,\infty)$;
\item $\MR_3(d)$ is non-empty if and only if $d=\frac{1}{2}$;
\item $\MR_4(d)$ is non-empty if and only if $d=1$.
\end{itemize}
\end{observation}

For dealing with $\S\in\MR_n(d)$ for a general $n$, let us introduce a notion of the \emph{standard form}:
\begin{definition}\label{standard}
We say that a matrix $\S\in\MR_n(d)$ is in the \emph{standard form} if
\begin{equation}\label{hat S}
\S=\frac{1}{\sqrt{d^2+n-1}}\left(\begin{array}{cccc|cccc}
+d & -1 & \cdots & -1 &  &  &  & \\
-1 & \ddots & &  & & \\
\vdots &  & \ddots & \\
-1 &  & & +d \\
\hline
 & & & & -d & +1 & \cdots & +1 \\
 & & & & +1 & \ddots & & \\
 & & & & \vdots & & \ddots & \\
 & & & & +1 & & & -d
\end{array}\right)
\begin{array}{cl}
\left.\begin{array}{c}
\vspace{7pt}
 \\
 \\
 \\
 \\
\end{array}\right\} & p
\vspace{5pt}\\
\left.\begin{array}{c}
\vspace{7pt}
 \\
 \\
 \\
 \\
\end{array}\right\} & n-p
\end{array}
\end{equation}
and $p\geq\frac{n}{2}$. For any $\S$ in the standard form, we define $Q=\sqrt{d^2+n-1}\S$, and denote the blocks of $Q$ by $Q^{(I)},Q^{(II)},Q^{(III)},Q^{(IV)}$, where $Q^{(I)}$ is the left upper one of size $p\times p$, i.e.,
$$
\S=\frac{1}{\sqrt{d^2+n-1}}\left(\begin{array}{c|c}
Q^{(I)} & Q^{(II)} \\
\hline
Q^{(III)} & Q^{(IV)}
\end{array}\right)\,.
$$

If $\S,\hat{\S}\in\MR_n(d)$, $\S\sim\hat{\S}$ and the matrix $\hat{\S}$ is in the standard form, we say that $\hat{\S}$ is a standard form of $\S$.
\end{definition}

\begin{remark}\label{StandardExists}
Evidently, for any $\S\in\MR_d(n)$ there exists its standard form $\hat{\S}\sim\S$; on the other hand, such an $\hat{\S}$ is generally not unique.
\end{remark}

\begin{lemma}\label{3rows}
Let $\S\in\MR_n(d)$ be in the standard form~\eqref{hat S} and let $p\geq3$.
\begin{itemize}
\item[(i)] If there exist $j,k\in\{2,3,\ldots,p\}$, $j\neq k$, such that $Q^{(I)}_{jk}=+1$, then $n+2d-2\equiv0\pmod4$ and $n-6d-6\geq0$.
\item[(ii)] If there exist $j,k\in\{2,3,\ldots,p\}$, $j\neq k$, such that $Q^{(I)}_{jk}=-1$, then $n-2d-2\equiv0\pmod4$.
\end{itemize}
\end{lemma}

\begin{proof}
The first $p$ rows of $Q$ form the matrix
$$
\left(\begin{array}{c|c}
Q^{(I)} & Q^{(II)}
\end{array}\right)
=\left(\begin{array}{cccc|cccc}
+d & -1 & \cdots & -1 & \pm1 & \pm1 & \cdots & \pm1 \\
-1 & \ddots & &  & & \\
\vdots &  & \ddots & \\
-1 &  & & +d \\
\end{array}\right)\,;
$$
since the rows of $(Q^{(I)}|Q^{(II)})$ are multiples of the rows of $\S$, they are mutually orthogonal.

For all $j\in\{p+1,\ldots,n\}$, let us multiply the $j$-th column of $(Q^{(I)}|Q^{(II)})$ by $-Q_{1j}$, which turns all the entries on the first row of $Q^{(II)}$ into $-1$.

\textit{(i)}\quad Let $Q^{(I)}_{jk}=+1$ for certain $j,k\in\{2,3,\ldots,p\}$, $j\neq k$. In this case we apply the following two transpositions simultaneously to rows and columns of $(Q^{(I)}|Q^{(II)})$: $2\leftrightarrow j$, $3\leftrightarrow k$. Note that this operation does not affect the orthogonality of the rows. As a result, the first three rows are
\begin{equation}
\begin{array}{ccccccc}
d & -1 & -1 & -1 \cdots -1 & -1 \cdots -1 & -1 \cdots -1 & -1 \cdots -1 \\
-1 & d & (+1) & +1 \cdots +1 & +1 \cdots +1 & -1 \cdots -1 & -1 \cdots -1 \\
-1 & (+1) & d & \underbrace{+1 \cdots +1}_{\ell_1} & \underbrace{-1 \cdots -1}_{\ell_2} & \underbrace{+1 \cdots +1}_{\ell_3} & \underbrace{-1 \cdots -1}_{\ell_4}
\end{array}
\end{equation}
They are orthogonal vectors from $\R^{1,n}$, hence these four equations must be fulfilled:
\begin{align}
3+\ell_1+\ell_2+\ell_3+\ell_4=n & \quad \text{(the vectors have $n$ components)} \label{sum+}\\
-2d-1-\ell_1-\ell_2+\ell_3+\ell_4=0 & \quad \text{(row 1 and row 2 are orthogonal)} \label{OG12+} \\
-2d-1-\ell_1+\ell_2-\ell_3+\ell_4=0 & \quad \text{(row 1 and row 3 are orthogonal)} \label{OG13+} \\
1+2d+\ell_1-\ell_2-\ell_3+\ell_4=0 & \quad \text{(row 2 and row 3 are orthogonal)} \label{OG23+}
\end{align}
We sum up all the four equations to obtain $2-2d+4\ell_4=n$, and from \eqref{sum+}$+$\eqref{OG23+}$-$\eqref{OG12+}$-$\eqref{OG13+} we get $6+6d+4\ell_1=n$. Since $\ell_1\in\N_0$ and $\ell_4\in\N_0$, it holds
$$
n+2d-2\equiv0\pmod4 \qquad\text{and}\qquad n-6d-6\geq0\,.
$$

\textit{(ii)}\quad Let $Q^{(I)}_{jk}=-1$ for certain $j,k\in\{2,3,\ldots,p\}$, $j\neq k$. Similarly as in the part (i), we apply the transpositions $2\leftrightarrow j$, $3\leftrightarrow k$ simultaneously to rows and columns of $(Q^{(I)}|Q^{(II)})$ to rearrange the first three rows into the form
\begin{equation}
\begin{array}{ccccccc}
d & -1 & -1 & -1 \cdots -1 & -1 \cdots -1 & -1 \cdots -1 & -1 \cdots -1 \\
-1 & d & (-1) & +1 \cdots +1 & +1 \cdots +1 & -1 \cdots -1 & -1 \cdots -1 \\
-1 & (-1) & d & \underbrace{+1 \cdots +1}_{\ell_1} & \underbrace{-1 \cdots -1}_{\ell_2} & \underbrace{+1 \cdots +1}_{\ell_3} & \underbrace{-1 \cdots -1}_{\ell_4}
\end{array}
\end{equation}
In the same way as above, we obtain equations $6-6d+4\ell_4=n$ and $2+2d+4\ell_1=n$, hence
$$
n-2d-2\equiv0\pmod4\,.
$$
\end{proof}

\begin{theorem}\label{real:n,d}
If $\S\in\MR_n(d)$ for $d<\frac{n}{2}-1$, then
\begin{itemize}
\item $n$ is even and $n\geq6$,
\item $d\in\N_0$,
\item $\frac{n}{2}+d$ is odd.
\end{itemize}
\end{theorem}

\begin{proof}
Let $\S\in\MR_n(d)$ for $d<\frac{n}{2}-1$. Then, with regard to Observation~\ref{small n}, we have $n\geq5$. We may assume without loss of generality (cf. Rem.~\ref{StandardExists}) that $\S$ is in the standard form. Therefore $p\geq\frac{n}{2}$ (see Def.~\ref{standard}), hence $p\geq 3$, which allows us to use Lemma~\ref{3rows}. Let $Q$, $Q^{(I)},Q^{(II)}$ have the meaning introduced in Definition~\ref{standard}. We divide the explanation into three alternatives:

\begin{itemize}
\item (The ``positive'' case.) Let us assume at first that $Q^{(I)}_{jk}=+1$ for all $j,k\in\{2,3,\ldots,p\}$, $j\neq k$. The orthogonality of the first two rows of $\S$ gives the condition
$$
-2d-(p-2)+\sum_{j=p+1}^{n}Q_{1j}Q_{2j}=0\,.
$$
However, since $\sum_{j=p+1}^{n}Q_{1j}Q_{2j}\leq n-p$ and at the same time it is assumed $p\geq\frac{n}{2}$, the condition cannot be satisfied for any $d<\frac{n}{2}-1$.
Consequently, the ``positive'' case is not possible.

\item (The ``mixed'' case.) Let there exist $j,k,j',k'\in\{2,3,\ldots,p\}$, $j\neq k$, $j'\neq k'$ such that $Q^{(I)}_{jk}=+1$ and $Q^{(I)}_{j'k'}=-1$. Then both statements (i) and (ii) of Lemma~\ref{3rows} apply, whence we get
$$
n+2d-2\equiv0\pmod4 \qquad\text{and}\qquad n-2d-2\equiv0\pmod4\,.
$$
The first condition, being equivalent to $\frac{n}{2}+d\equiv1\pmod2$, means that $\frac{n}{2}+d$ is odd. Moreover, together with the second condition, it implies $2n-4\equiv0\pmod4$ and $4d\equiv0\pmod4$, hence $n$ is even and $d$ is integer.

\item (The ``negative'' case.) Let finally $Q^{(I)}_{jk}=-1$ for all $j,k\in\{2,3,\ldots,p\}$, $j\neq k$. Here we distinguish two situations:
\begin{itemize}
\item If $p=n$, we have $\S=\frac{1}{\sqrt{d^2+n-1}}\left((d+1)I_n-J_n\right)$. In this case, the orthogonality of the first two rows requires $-2d+n-2=0$, hence $d=\frac{n}{2}-1$, which contradicts our assumption $d<\frac{n}{2}-1$.
\item If $p<n$, then the orthogonality of the $1$st row and the $(p+1)$-st row of $Q$ leads to the condition
$$
d\cdot Q_{p+1,1}-\sum_{j=2}^{p}Q_{p+1,j}+Q_{1,p+1}\cdot(-d)+\sum_{j=p+1}^{n}Q_{1,j}=0\,.
$$
Since $Q_{p+1,1}=Q_{1,p+1}$, the terms with $d$ cancel. The remaining condition is of the type $\underbrace{1+1+\cdots+1}_{\ell}\underbrace{-1-1\cdots-1}_{n-2-\ell}=0$ for a certain $\ell\in\N_0$, and thus can be satisfied only when $n$ is even. Using this fact together with the relation $n-2d-2\equiv0\pmod4$ from Lemma~\ref{3rows}~(ii), we obtain $d\in\N_0$ and $\frac{n}{2}-d\equiv1\pmod2$, which is equivalent to $\frac{n}{2}+d\equiv1\pmod2$.
\end{itemize}
\end{itemize}

Finally, the inequality $n\geq6$ follows from $n\geq5$ and from the even parity of $n$, which has been proved above.
\end{proof}

\begin{remark}
It follows from Theorem~\ref{real:n,d}:
\begin{itemize}
\item If $n$ is odd, necessarily $d=\frac{n}{2}-1$.
\item If $n\equiv2\pmod4$, then $d\in\left\{0,2,4,\ldots,\frac{n}{2}-1\right\}$.
\item If $n\equiv0\pmod4$, then $d\in\left\{1,3,5,\ldots,\frac{n}{2}-1\right\}$.
\end{itemize}
\end{remark}

It turns out that for certain values of $d$, a more detailed description of $\S$ can be found. Let us start with the following observation.
\begin{observation}\label{e.v. J}
The matrix $J_k$ has a simple eigenvalue $k$ corresponding to the eigenvector $\vec{w}:=(1,1,\ldots,1)^T$, and the eigenvalue $0$ of multiplicity $k-1$ corresponding to the eigenspace $\vec{w}^\bot$.
\end{observation}

Observation~\ref{e.v. J} will help us to characterize $\S\in\MR_n(d)$ for $d$ exceeding $\frac{n}{6}-1$:
\begin{proposition}\label{d>n/6-1}
\textit{(i)}\quad Let $\S\in\MR_n(d)$ for $d\in\left(\frac{n}{6}-1,\frac{n}{2}-1\right)$. Then $p=\frac{n}{2}$ and there is a $G\in\{-1,1\}^{p,p}$ such that
\begin{equation}\label{great d}
\S\sim\frac{1}{\sqrt{d^2+n-1}}\left(\begin{array}{cc}
(d+1)I_p-J_p & G \\
G^T & -(d+1)I_p+J_p
\end{array}\right)\,.
\end{equation}
The matrix $G$ has these properties: $G$ is normal, $G$ commutes with $J_p$, and $GG^T=(n-2d-2)I_p+(2d+2-n/2)J_p$.\\
\textit{(ii)}\quad On the other hand, if $d$ and $G$ fulfil the conditions above, then any $\S$ satisfying~\eqref{great d} belongs to $\MR_n(d)$.
\end{proposition}

\begin{proof}
We assume without loss of generality that $\S$ is in the standard form. Since $n\geq6$ according to Theorem~\ref{real:n,d}, it holds $p\geq\frac{n}{2}\geq3$. Let $Q^{(I)},Q^{(II)},Q^{(III)},Q^{(IV)}$ have the meaning introduced in Definition~\ref{standard}. The proof will be carried out in five steps.

\emph{Step 1.}\quad Since $p\geq3$ and $d>\frac{n}{6}-1$, it immediately follows from Lemma~\ref{3rows}~(i) that $Q^{(I)}=(d+1)I_p-J_p$.

\emph{Step 2.}\quad We prove that $m\geq p$.

Since $Q^{(I)}=(d+1)I_p-J_p$, it holds
$$
m=\rank(\S+I)\geq\rank\left(\frac{1}{\sqrt{d^2+n-1}}Q^{(I)}+I_p\right)=\rank\left(\left(d+1+\sqrt{d^2+n-1}\right)I_p-J_p\right)\,.
$$
Our goal is to show that the matrix $L:=\left(d+1+\sqrt{d^2+n-1}\right)I_p-J_p$ is invertible (equivalently, it does not have the eigenvalue $0$).

Since the eigenvalues of $J_p$ are $0$ and $p$ according to Observation~\ref{e.v. J}, the eigenvalues of $L$ are $d+1+\sqrt{d^2+n-1}$ and $d+1+\sqrt{d^2+n-1}-p$. The former one is trivially nonzero, thus it suffices to show $d+1+\sqrt{d^2+n-1}-p\neq0$. We proceed by contradiction:

Let $d+1+\sqrt{d^2+n-1}-p=0$. Then the value $\ell:=\sqrt{d^2+n-1}-d=p-1-2d$ satisfies $\ell\in\N$ (because $p\in\N$ and $\d\in\N$, cf. Thm.~\ref{real:n,d}), and also
$$
\ell=\sqrt{d^2+n-1}-d=\frac{n-1}{\sqrt{d^2+n-1}+d}=\frac{n-1}{p-1}\leq\frac{n-1}{\frac{n}{2}-1}<3\,,
$$
because $n\geq6$. Consequently, $\ell=1$ or $\ell=2$.
\begin{itemize}
\item Case $\ell=1$ implies $\sqrt{d^2+n-1}-d=1$, hence $d=\frac{n}{2}-1$, which contradicts the assumption $\frac{n}{6}-1<d<\frac{n}{2}-1$.
\item Case $\ell=2$ implies $\sqrt{d^2+n-1}-d=2$, hence $d=\frac{n-5}{4}$. However, since $n$ is even (see Thm.~\ref{real:n,d}), such $d$ is not an integer, and thus is not admissible.
\end{itemize}
Therefore $L$ is invertible, hence $\rank(L)=p$, thus indeed $m=\rank(L)\geq p$.

\emph{Step 3.}\quad
Since $p\geq\frac{n}{2}$ and at the same time $m\geq p$ due to Step 2, Theorem~\ref{m,p,d} implies $m=p=\frac{n}{2}$.

\emph{Step 4.}\quad We show that $Q^{(IV)}=-(d+1)I_{n-p}+J_{n-p}$.

Since $p=\frac{n}{2}$ and $\S$ is in the standard form, it is obvious that the matrix $\S':=-P\S P^{-1}$ for $P=\left(\begin{array}{cc}
0 & I_p \\
I_p & 0
\end{array}\right)$, which takes the form
$\S'=\frac{1}{\sqrt{d^2+n-1}}\left(\begin{array}{c|c}
-Q^{(IV)} & -Q^{(III)} \\
\hline
-Q^{(II)} & -Q^{(I)}
\end{array}\right)$, satisfies $\S'\sim\S$ and is in the standard form as well.
Therefore, with regard to Step 1, it holds $-Q^{(IV)}=(d+1)I_p-J_p$.

\emph{Step 5.}\quad The result of Step 1 together with the hermiticity of $\S$ imply \eqref{great d}. 
From the unitarity of $\S$, $\S\S^*=I_n$, we immediately obtain the properties $GG^T=G^TG$, $GJ_p=J_pG$ and $GG^T=(n-2d-2)I_p+(2d+2-n/2)J_p$. And vice versa, if $G$ fulfils these conditions and $d\in\left(\frac{n}{6}-1,\frac{n}{2}-1\right)$, then the matrix~\eqref{great d} satisfies $\S\S^*=I$.
\end{proof}

Proposition~\ref{d>n/6-1} allows us to formulate the necessary and sufficient condition of the existence of $\S\in\MR_n(d)$ for all $d>\frac{n}{6}-1$:
\begin{theorem}\label{design}
\begin{itemize}
\item[(i)] There is no $\S\in\MR_n(d)$ for $d\in\left(\frac{n}{6}-1,\frac{n}{4}-\frac{3}{2}\right)$.
\item[(ii)] An $\S\in\MR_n(d)$ for $d\in\left[\frac{n}{4}-\frac{3}{2},\frac{n}{2}-1\right)$ exists if and only if the value
$$
q:=\sqrt{\frac{n}{2}+\left(\frac{n}{2}-1\right)\left(2d+2-\frac{n}{2}\right)}
$$
is an integer and there exists a symmetric $\left(\frac{n}{2},k,\lambda\right)$-design for $k=\frac{n}{4}-\frac{q}{2}$ and $\lambda=\frac{1}{2}(d-q+1)$.
\end{itemize}
\end{theorem}
\begin{proof}
\textit{(i)}\  According to Proposition~\ref{d>n/6-1}, an $\S\in\MR_n(d)$ for $d>\frac{n}{6}-1$ exists if and only if there exists a normal $G\in\{-1,1\}^{p,p}$ ($p=\frac{n}{2}$) satisfying
$$
GJ_p=J_pG\,,\qquad GG^T=(n-2d-2)I_p+(2d+2-n/2)J_p\,.
$$
With regard to Observation~\ref{e.v. J}, the matrix $GG^T$ has a simple eigenvalue $\frac{n}{2}+(\frac{n}{2}-1)(2d+2-\frac{n}{2})$ and a corresponding eigenvector $\vec{w}=(1,1,\ldots,1)^T$. Since $GG^T$ is a nonnegative matrix, necessarily $\frac{n}{2}+(\frac{n}{2}-1)(2d+2-\frac{n}{2})\geq0$, hence we obtain the condition
$$
d\geq\frac{n}{4}-\frac{3}{2}-\frac{1}{n-2}\,.
$$
Finally, we know from Theorem~\ref{real:n,d} that $d$ is integer, $n$ is even and $n\geq6$, for this reason the last condition can be equivalently written as $d\geq\frac{n}{4}-\frac{3}{2}$. Hence (i) is proved.

In the proof of (ii), we start from the following obvious statement:\\
\emph{If $M$ is normal, then $\vec{v}$ is an eigenvector of $M$ with a simple eigenvalue $\sigma$ if and only if $\vec{v}$ is an eigenvector of $MM^*$ with a simple eigenvalue $|\sigma|^2$.} \\
We have found in part (i) that $\vec{w}=(1,1,\ldots,1)^T$ is an eigenvector of $GG^T$ corresponding to a simple eigenvalue $\frac{n}{2}+(\frac{n}{2}-1)(2d+2-\frac{n}{2})\geq0$. Therefore, due to the above statement, it is at the same time an eigenvector of $G$ corresponding to a simple eigenvalue $\mu$ of modulus $q:=\sqrt{\frac{n}{2}+\left(\frac{n}{2}-1\right)\left(2d+2-\frac{n}{2}\right)}$. Since both $\vec{w}$ and $G$ are real and their entries are integers $\pm1$, the eigenvalue $\mu$ must be real and integer, hence $q\in\N_0$ and $\mu=\pm q$.
The actual sign of $\mu$ is irrelevant, because we can always turn $G$ in \eqref{great d} into $-G$ by multiplying the rows $\frac{n}{2}+1,\ldots,n$ of $\S$ and the columns $\frac{n}{2}+1,\ldots,n$ of $\S$ by $-1$. Let us assume for definiteness $\mu=-q$. We define $A=\frac{1}{2}\left(G+J_p\right)$; then $A\in\{0,1\}^{p,p}$ and
$$
AJ_p=\underbrace{\frac{1}{2}\left(-q+\frac{n}{2}\right)}_{k}J_p\,,\qquad AA^T=\underbrace{\left(\frac{n}{4}-\frac{d}{2}-\frac{1}{2}\right)}_{k-\lambda}I_p+\underbrace{\frac{d-q+1}{2}}_{\lambda}J_p\,.
$$
It follows (see~\eqref{IncidMat}) that $A$ is the incidence matrix of a symmetric $\left(\frac{n}{2},k,\lambda\right)$-design.
\end{proof}

\begin{remark}
After finishing this work we found out about paper~\cite{La83} in which orthogonal (not necessarily symmetric) MPS matrices with constant integral diagonal have been studied. It follows from there that our Theorem~\ref{real:n,d} is valid even when the symmetry of the matrix is weakened to the non-skew-symmetry.
\end{remark}

Let us finish the section by a series of remarks on construction of matrices $\S\in\MR_n(d)$.

\subsection*{Notes on constructions of $\S\in\MR_n(d)$}

\begin{itemize}

\item \emph{An $\S\in\MR_n(\frac{n}{2}-1)$ exists for all $n\geq2$.} For example $\S=I_n-\frac{2}{n}J_n$, see also Section~\ref{Construction}.

\item \emph{An $\S\in\MR_n(\frac{n}{2}-3)$ exists for any even $n$.} To obtain $\S$, set $\alpha=\pi$ in~\eqref{UpperInt}.

\item \emph{An $\S\in\MR_n(d)$ for $d\in\left[\frac{n}{4}-\frac{3}{2},\frac{n}{2}-3\right)$ exists if and only if $q:=\sqrt{\frac{n}{2}+\left(\frac{n}{2}-1\right)\left(2d+2-\frac{n}{2}\right)}\in\N_0$ and there exists a symmetric $\left(\frac{n}{2},k,\lambda\right)$-design for $k=\frac{n}{4}-\frac{q}{2}$ and $\lambda=\frac{1}{2}(d-q+1)$.} (Cf. Thm.~\ref{design}.)

The construction follows from Theorem~\ref{design}. If $A$ is the incidence matrix of a symmetric $\left(\frac{n}{2},\frac{n}{4}-\frac{q}{2},\frac{d}{2}-\frac{q}{2}+\frac{1}{2}\right)$-design, we set $G=2A-J_{\frac{n}{2}}$ and construct $\S$ according to~\eqref{great d}.

\item In particular, for $d=\frac{n}{4}-\frac{3}{2}$ we obtain:
\begin{proposition}\label{Hadam.}
There exists an $\S\in\MR_n(\frac{n}{4}-\frac{3}{2})$ if and only if there exists a real Hadamard matrix of order $\frac{n}{2}+1$.
\end{proposition}
\begin{proof}
We apply Theorem~\ref{design}. If $d=\frac{n}{4}-\frac{3}{2}$, then $q=1\in\N_0$ and $k=\frac{n}{4}-\frac{1}{2}=\frac{1}{2}\left(\frac{n}{2}+1\right)-1$, $\lambda=\frac{n}{8}-\frac{3}{4}=\frac{1}{4}\left(\frac{n}{2}+1\right)-1$. Therefore the existence of an $\S\in\MR_n(\frac{n}{4}-\frac{3}{2})$ is equivalent to the existence of an 
$(N-1,\frac{1}{2}N-1,\frac{1}{4}N-1)$-design where $N=\frac{n}{2}+1$, which is further equivalent to the existence of an Hadamard matrix of order $N$ (see~\cite{BJL}, Lemma I.9.3).
\end{proof}
A direct proof of Proposition~\ref{Hadam.} is the following:\\
The validity of implication $\Leftarrow$ is confirmed by the construction \eqref{S from core}, therefore it suffices to prove $\Rightarrow$. Let $\S\in\MR_n(\frac{n}{4}-\frac{3}{2})$. Due to Proposition~\ref{d>n/6-1}, $\S$ satisfies \eqref{great d} for a certain $G\in\{-1,1\}^{\frac{n}{2},\frac{n}{2}}$. In the same way as in the proof of Proposition~\ref{d>n/6-1}, we can show that $\vec{w}=(1,1,\ldots,1)^T$ is an eigenvector of $G$ corresponding to a simple eigenvalue $\mu=\pm1$. 
We define
$$
H:=\left(\begin{array}{c|ccc}
-\mu & 1 & \cdots & 1 \\
\hline
1 & \\
\vdots & & G \\
1 & 
\end{array}\right)\,,
$$
and using the properties of $G$ derived in Proposition~\ref{d>n/6-1}, we show that
$HH^*=(\frac{n}{2}+1)I$, thus $H$ is an Hadamard matrix.

\emph{Remark.}\ 
It follows from Theorem~\ref{real:n,d} that an $\S\in\MR_n(\frac{n}{4}-\frac{3}{2})$ exists \emph{only if} $n\equiv6\pmod8$. Proposition~\ref{Hadam.} implies that \emph{if the Hadamard conjecture is true, then an $\S\in\MR_n(\frac{n}{4}-\frac{3}{2})$ exists if and only if $n\equiv6\pmod8$}.

\item \emph{If $d\in\left(\frac{n}{6}-1,\frac{n}{4}-\frac{3}{2}\right)$, then $\MR_n(d)$ is empty.} (Thm.~\ref{design}.)

\item \emph{If there is a symmetric conference matrix $C$ of order $\frac{n}{2}$, then there exists an $\S\in\MR_n(1)$.} To obtain $\S$, set $\alpha=0$ in~\eqref{S C}.

\item The existence of $\S\in\MR_n(0)$ is trivially equivalent to the existence of a symmetric conference matrix of order $n$.

\item For constructions of $\S\in\MR_n(d)$ with $p=n$, we refer to~\cite{SL82}, where real symmetric MPS matrices with constant diagonal were studied and certain methods of their construction have been proposed.

\end{itemize}

%%%%%%%%%%%%%%%%%%%%%%%%%%%%%%%%%%%%%%%%%%%%%%%%%%%%%%%%%%%%%%%%%%%%%%%%%%%%%%%%%%%%

\section{An application: Quantum graphs}\label{QuantumGraphs}

Let us briefly explain in what context matrices from the sets $\MC_n(d)$ emerge in quantum mechanics on graphs.

Consider a metric graph, i.e., a set of vertices and a set of edges, the edges connect the vertices, each edge has a given length. Let us suppose that the graph is of microscopic size and that there is a particle, for example an electron, having certain energy and moving along the graph edges. 
As the size of the system is very small, the behaviour 
of the particle is governed by the laws of quantum mechanics. In particular, its position cannot be exactly determined, one can only find the probability density of its occurrence in a given point $x$ of the graph, which is given as $|\Psi(x)|^2$, where $\Psi$ is the \emph{wave function} of the particle. The function $\Psi$ depends on the topology of the graph, on the lengths of the edges, on the particle energy $E$ and on \emph{physical characteristics of the vertices (junctions)}. The physical characteristics of each junction are expressed by the 
\emph{scattering matrix} $\S$ that has the following properties~\cite{Ne66}:
\begin{itemize}
\item $\S$ is a complex $n\times n$ matrix, where $n$ is the vertex degree.
\item Let the edges coupled at the junction be numbered by $1,\ldots,n$. If the quantum particle comes in the junction from the $j$-th line, then it is scattered into all lines $1,\ldots,n$ (including the $j$-th line itself) with the probabilities $|\S_{1j}|^2,\ldots,|\S_{nj}|^2$. In other words, the squared moduli of the entries of $\S$ correspond to the scattering probabilities at the junction.
\item $\S$ is always unitary (this property may be viewed as the quantum version of Kirchhoff's law, see~\cite{KS99,KS00}).
\item $\S$ generally depends on energy (where $\S(E)$ can be uniquely calculated from $\S(1)$).
\item $\S$ is energy-independent if and only if $\S$ is Hermitian.
\end{itemize}
Now it is obvious what role the Hermitian unitary MPS matrices play in this theory. Consider a junction of degree $n$. If its physical characteristics are described by a Hermitian unitary MPS matrix $\S\in\MC_n(d)$, then the particle is transmitted from any edge to any other edge with equal probabilities, also the reflection probabilities are the same at all edges, and furthermore, the probabilities are independent of the particle energy. The parameter $d$ squared represents the ratio of the reflection probability to the probability that the particle is transmitted to any chosen edge different from the incoming one.

The existence of an $\S\in\MC_n(d)$ thus determines whether or not it is possible to physically construct (manufacture) an equally-transmitting junction with the given scattering ratio $d^2$.
It is also noteworthy that real scattering matrices $\S$, examined in Section~\ref{Real}, correspond to junctions with the additional physical property of time reversibility.

Let us add that quantum graphs serve as efficient models of realistic physical systems where a quantum particle moves along thin paths.
The concept has been originally introduced as a tool for calculating the spectra of aromatic hydrocarbons~\cite{RS53}. Nowadays it is widely used for the study of systems where an electron is confined to thin nano-sized wires or networks made for example of semiconductors. It also allows to examine physical properties of carbon nano-structures, including graphene and single-wall nano-tubes~\cite{KP07}. With regard to the current rapid progress in nanotechnologies, quantum mechanics on graphs attracts recently a lot of attention and the literature is very extensive, let us refer for instance to the proceedings~\cite{EKST08}.

The particular problem of equal transmission probabilities was considered for the first time in the paper~\cite{HSW07}. Its authors studied quantum junctions described by unitary (but not necessarily Hermitian) $n\times n$ scattering matrices with the property $\S(1)_{jj}=0$, $|\S(1)_{j\ell}|=1/\sqrt{n-1}$. The motivation originated mainly in the analysis of the so-called trace formula.

\section*{Acknowledgements}
We thank Prof.~L\'aszl\'o Feh\'er and Prof.~Izumi Tsutsui for stimulating discussions.
This research was supported by the Japanese Ministry of Education, Culture, Sports, Science and Technology under the Grant number 21540402.

%%%%%%%%%%%%%%%%%%%%%%%%%%

\end{document}